\documentclass[prd,11pt,tightenlines,nofootinbib,superscriptaddress]{revtex4}
\usepackage{amsfonts,amssymb,amsthm,bbm,amsmath}

\newtheorem{prop}{Proposition}
\newtheorem{defi}{Definition}

\newcommand{\C}{{\mathbb C}}
\newcommand{\N}{{\mathbb N}}
\newcommand{\R}{{\mathbb R}}

\newcommand{\cE}{{\mathcal E}}

\newcommand{\cH}{{\mathcal H}}

\newcommand{\cC}{{\mathcal C}}

\newcommand{\SU}{\mathrm{SU}}
\newcommand{\SL}{\mathrm{SL}}
\newcommand{\GL}{\mathrm{GL}}
\newcommand{\U}{\mathrm{U}}

\newcommand{\id}{\mathbbm{1}}

\newcommand{\vJ}{\vec{J}}

\newcommand{\be}{\begin{equation}}
\newcommand{\ee}{\end{equation}}
\newcommand{\beq}{\begin{eqnarray}}
\newcommand{\eeq}{\end{eqnarray}}
\newcommand{\bea}{\begin{eqnarray}}
\newcommand{\eea}{\end{eqnarray}}
\newcommand{\nn}{\nonumber}

\newcommand{\matri} [2] {\left ( \begin{array}{#1}#2\end{array} \right ) }

\newcommand{\bin} [2] {\left (\begin{array}{c}#2\\#1\end{array} \right ) }

\newcommand{\su}{{\mathfrak su}}
\renewcommand{\sl}{{\mathfrak sl}}

\renewcommand{\u}{{\mathfrak u}}
\newcommand{\bra}{\langle}
\newcommand{\ket}{\rangle}
\newcommand{\la}{\langle}
\newcommand{\ra}{\rangle}

\newcommand{\tr}{{\mathrm Tr}}
\newcommand{\f}{\frac}

\newcommand{\rd}{\mathrm{d}}

\newcommand{\Ref}[1]{(\ref{#1})}
\newcommand{\Bz}{{\bf z}}
\newcommand{\Bw}{{\bf w}}

\newcommand{\zo}{z^{0}}
\newcommand{\zi}{z^{1}}
\newcommand{\wo}{w^{0}}
\newcommand{\wi}{w^{1}}
\newcommand{\bzo}{\bar{z}^{0}}
\newcommand{\bzi}{\bar{z}^{1}}

\def\tz{{\widetilde{z}}}

\def\nn{\nonumber}

\def\arr{\rightarrow}

\begin{document}

\title{$\U(N)$ Coherent States for Loop Quantum Gravity}

\author{{\bf Laurent Freidel}\email{lfreidel@perimeterinsititute.ca}}
\affiliation{ Perimeter Institute for Theoretical Physics,
Waterloo, N2L-2Y5, Ontario, Canada.}
\author{{\bf Etera R. Livine}\email{etera.livine@ens-lyon.fr}}
\affiliation{Laboratoire de Physique, ENS Lyon, CNRS-UMR 5672, 46 All\'ee d'Italie, Lyon 69007, France.}

\date{July 2009}

\begin{abstract}
We investigate the geometry of the space of $N$-valent $\SU(2)$ intertwiners.
We propose a new set of holomorphic operators acting on this space and a new set of coherent states which are covariant under $\U(N)$ transformations.
These states are labeled by elements of the Grassmannian $Gr_{N,2}$, they possess a direct geometrical interpretation in terms of framed polyhedra  and are shown to be related to the well-known coherent intertwiners.
\end{abstract}

\maketitle


\section{Introduction}

Loop quantum gravity is a tentative canonical quantization of general relativity where the quantum states of geometry are the so-called spin network states. A spin network is based on a graph $\Gamma$ dressed up  with half-integer spins $j_e$ on its edges and intertwiners $i_v$ on its vertices. The spins define quanta of area while the intertwiners describe chunks of (space) volume. The dynamics then acts on the spins $j_e$ and intertwiners $i_v$, and can also deform the underlying graph $\Gamma$.

Here, we would like to focus on the structure of the space of intertwiners describing the chunk of space.
We focus our study on a region associated with a single vertex of a graph and arbitrary high valency.
Associated with this setting there is, as discussed in  \cite{UN}, a classical geometrical description.
To each edge going out of this vertex we can associate a dual surface element or face whose area is given by the spin label.
The collection of these faces encloses a 3-dimensional volume whose boundary  forms a 2-dimensional polygon with the topology of a sphere.
This 2d-polygon is such that each vertex is trivalent.
At the quantum level the choice of intertwiner $i_v$ attached to the vertex is supposed to describe the shape of the full dual surface and give the volume contained in that surface.

In previous work \cite{OH,UN}, it was uncovered that the space of ($N$-valent) intertwiners carry an {\it irreducible} representation of the unitary group $\U(N)$. These irreducible representations of $\U(N)$ are labeled by one integer: the total area of the dual surface (defined as the sum of the spins coming through this surface). Then the $\U(N)$ transformations deform of the shape of the intertwiner at fixed area.
Moreover, it was proposed in \cite{UN} to view $N$-valent intertwiners as functions on $\U(N)$.
More precisely, it was hinted 
that the space of labels  associated with $N$-valent intertwiners is given by a Grassmannian space defined as a coset of $\U(N)$.
This $\U(N)$-coset space was then interpreted as the space of polyhedra with $N$-faces (and only trivalent vertices) at fixed total area, times one $\U(1)$ phase per face. These phases can be interpreted as a choice of 2d-frame on each face of the polyhedron This provides a clean geometric interpretation to the space of intertwiners as wave-functions over the space of classical $N$-faced polyhedron. It also leads to a clearer picture of what the discrete surface dual to the intertwiner should look like in the semi-classical regime.

Here, we would like to go further: we construct a coherent state basis for the space of $N$-intertwiners at fixed area and we show that
 the Hilbert space of $N$-valent intertwiners is isomorphic to the space of {\it holomorphic} functions on the Grassmanian of complex two planes in $\C^{N}$.

 In order to do so we introduce creation and annihilation operators, which are still $\SU(2)$-invariant and thus act on the space of intertwiners, but do not commute with the total area. These new operators allow to move from one $\U(N)$ irreducible representation to another, thus endowing the full intertwiner space with a Fock space interpretation where the ``number of particles" is now the total area.
We check that the $\u(N)$ generators $E_{ij}$ and the new annihilation/creation operators $F_{ij},F_{ij}^\dag$ form a closed algebra and we show how to use these new creation operators to build coherent states that transform consistently under the $\U(N)$ action.
This allow us to  show that the space of $N$-valent intertwiners is isomorphic to the space of holomorphic functions on the Grassmanian of complex two planes in $\C^{N}$.

In this work we present the explicit construction for these new coherent states which are covariant under $U(N)$,
we compute their norm, scalar product and show that they provide an overcomplete basis.
We compute their semi-classical expectation values and uncertainties and show that they are simply related to the Livine-Speziale coherent intertwiners \cite{coh1}, currently used in the construction of the Engle-Pereira-Rovelli-Livine (EPRL) and Freidel-Krasnov (FK) spinfoam models and their corresponding semi-classical boundary states \cite{EPRL,FK}.
These new coherent states allow us to confirm the polyhedron interpretation of the intertwiner space and in particular underline the relevance of the $\U(1)$ phase/frame attached to each face, which appears very similar to the extra phase entering the definition of the discrete twisted geometries for loop gravity \cite{twisted}.


We hope that this $\U(N)$ framework will turn out useful when studying the dynamics, symmetries and semi-classical regime of loop quantum gravity and spinfoam models.

\section{The $\U(N)$ Framework for $\SU(2)$ Intertwiner Space}

\subsection{The $\U(N)$ Action on Intertwiners}

We start by reviewing the framework introduced in \cite{OH,UN} for studying the space of $\SU(2)$ intertwiners relevant to the definition of the spin network states for quantum geometry.
In the following, we call $V^j$ the Hilbert space corresponding to the irreducible representation (irreps) of $\SU(2)$ with spin $j\in\N/2$. Its dimension is $d_j=(2j+1)\in\N$.
Considering $N$ irreps of spins $j_1,..,j_N\,\in\N/2$, the corresponding space of intertwiners with $N$ legs consists in the vectors in the tensor product of these representations which are invariant under the global $\SU(2)$ action:
\be
\cH_{j_1,..,j_N}\,\equiv\, \textrm{Inv}[V^{j_1}\otimes..\otimes V^{j_N}].
\ee
We further introduce the space of intertwiners with $N$ legs and fixed total are $J=\sum_i j_i$~:
\be
\cH_N^{(J)}\,\equiv\,\bigoplus_{\sum_i j_i=J}\cH_{j_1,..,j_N},
\ee
and finally the total space of $N$-valent intertwiners:
\be
\cH_N\,\equiv\,\bigoplus_{\{j_i\}}\cH_{j_1,..,j_N}\,=\, \bigoplus_{J\in\N}\cH_N^{(J)}.
\ee

The main result from the previous works \cite{OH,UN} is that the Hilbert spaces $\cH_N^{(J)}$ carry irreducible representations of $\U(N)$. In \cite{UN}, this was used to perform an explicit intertwiner counting related to black hole entropy calculations in loop quantum gravity. Here, we would like to push further the analysis of this $\U(N)$ structure in order to understand better in particular its geometric interpretation.

The original motivation to consider this $\U(N)$ action on the space of intertwiners was to identify a closed algebra of $\SU(2)$-invariant observables. More precisely, let us call
$\vJ_{i}$ for the generators in the $\su(2)$ Lie algebra associated to the $i$-th leg of the intertwiner. The usual invariant operators that are considered and used to probe the intertwiner space are the scalar product operators $\vJ_{i}\cdot \vJ_{j}$ for all couples of indices $(i,j)$.
%
%
The standard issue with these scalar product operators is that they do not form a closed Lie algebra: their commutator with each other will generate cubic operators in the $J$'s -the volume operators- and the next step will produce fourth order operators and so on. This tower of higher order operators leads to a a priori infinite dimensional algebra. This is problematic, as an example, when we try to build coherent semi-classical states, minimizing uncertainty relations, since we have an a priori infinite list of such relations to satisfy.
%
%
The solution proposed in \cite{OH}, and later studied in more details in \cite{UN}, is to construct new invariant operators, which turn out to form a closed $\u(N)$ algebra.

The basic tool entering this construction is the Schwinger representation of the $\su(2)$ algebra in term of a pair of (uncoupled) harmonic oscillators:
$$
[a,a^\dag]=[b,b^\dag]=1,\qquad [a,b]=0,
$$
and we define the $\su(2)$ generators as quadratic operators in the $a,b$'s:
\be
J^z\equiv\f12(a^\dag a-b^\dag b),\quad
J^+\equiv a^\dag b, \quad
J^-\equiv ab^\dag.
\ee
We also define the (half) total energy $E$ (or more precisely the number of quanta):
\be
\cE\equiv \f12(a^\dag a+b^\dag b)
\ee
It is straightforward to check that the commutation relations reproduce the expected $\su(2)$ structure:
$$
[J^z,J^\pm]=\pm J^\pm, \quad [J^+,J^-]=2J^z, \quad
[\cE,\vJ]=0.
$$
From these definitions, we identify the correspondence between the usual basis of the Hilbert space for harmonic oscillators and the standard basis of $\SU(2)$ irreps in term of the spin $j$ and magnetic momentum $m$~:
\be
|j,m\ra\,=|n_a,n_b)_{OH}=\frac{(a^{\dagger})^{n_{a}}}{\sqrt{n_{a}!}}\frac{(b^{\dagger})^{n_{b}}}{\sqrt{n_{b}!}} |0,0)_{OH},
\qquad\textrm{with}\quad
j=\f12(n_a+n_b),\quad m=\f12(n_a-n_b).
\ee
The (half) total energy $\cE$ gives the spin $j$ and we check that the $\su(2)$ Casimir operator can indeed be simply expressed in term of $\cE$:
$$
\cC=\vec{J}^2=\cE(\cE+1)=\left(\cE+\f12\right)^2- \f14.
$$

Now considering intertwiners with $N$ legs, we take $N$ irreps of $\SU(2)$, so we use $2N$ oscillators $a_i,b_i$. Following  \cite{OH}, we define the quadratic operators~\footnotemark acting on couples of punctures $(i,j)$:
\be
E_{ij}\equiv (a_i^\dag a_j+b_i^\dag b_j), \quad
E_{ij}^\dag=E_{ji}.
\ee
\footnotetext{
An alternative definition  of these $\u(N)$ generators is given by the symmetrized operators:
$$
\tilde{E}_{ij}\equiv \,\f12(\{a_i^\dag,a_j\}+\{b_i^\dag, b_j\}),
$$
which differ from the previous definition by $+1$ when $i=j$. This corresponds to a shift $j_i\arr j_i+\f12$ in the area spectrum. However this means that the $\tilde{E}$ operators do not all vanish on the vacuum state $|0)$.
}
%
These operators commute with the global $\SU(2)$ transformations generated by $\vJ=\sum_k \vJ_{k}$, so they legitimately define operators acting the intertwiner space:
\be
\forall i,j,\quad \left[\sum_k \vJ_{k}\,,\,E_{ij}\right]=0.
\ee
We check that these new operators form a closed $\u(N)$ Lie algebra:
\be
[E_{ij},E_{kl}]\,=\,
\delta_{jk}E_{il}-\delta_{il}E_{kj}.
\ee
The diagonal operators $E_i\equiv E_{ii}$ form the (abelian) Cartan sub-algebra. Their value on a state gives twice the spin on the $i$th leg, $2j_i$. The $\u(1)$ Casimir operator $E\equiv\sum_k E_k$ generates the global $\U(1)$ phase in $\U(N)$ and its value on a state gives twice the total area $2J=\sum_k 2j_k$. The off-diagonal operators $E_{ij}$ define the lowering and raising operators in $\u(N)$. At the level of spins, $E_{ij}$ increases the spin $j_i$ by $+\f12$ while decreasing the spin $j_j$ by a half-step, thus leaving invariant the total area $J=\sum_k j_k$.

We can easily link the $E_{ij}$ operators with the standard scalar product operators $\vJ_{i}\cdot \vJ_{j}$. The $E$'s are quadratic in the oscillators $a,b$'s while the $\vJ\cdot\vJ$ are quartic, and somehow we can interpret the $E$'s as taking the square-root of the scalar product operators. More precisely, as was shown in \cite{OH}, it is possible to express the operator $\vJ^{(i)}\cdot \vJ^{(j)}$ as a quadratic polynomial of $E_{ij}, E_{ji},E_i,E_j$~:
\be
\forall i\ne j,\,\vJ_{i}\cdot \vJ_{j}
\,=\,
\f12E_{ij} E_{ji} -\f14E_iE_j-\f12 E_i
\,=\,
\f12E_{ji} E_{ij} -\f14E_iE_j-\f12 E_j.
\ee

The $\U(N)$ transformations generated by the operators $E_{ij}$ describe  deformations of the boundary surface dual to the intertwiner vertex which preserves the total area. Having in mind previous work on the $N\arr\infty$ limit of the $\U(N)$ symmetry (e.g. \cite{hoppe}), we proposed in \cite{UN} to interpret these $\U(N)$ transformations as the area-preserving diffeomorphisms (which preserve the symplectic structure) on the discrete quantum sphere defined by the space of $N$-valent intertwiners.

Moreover, we identified in \cite{UN} the precise $\U(N)$ action on the intertwiner space. To this purpose, we realized that the particular realization of the $\u(N)$ generators in term of a couple of harmonic oscillators lead to a set of quadratic constraints:
\be
\forall i,\,
\sum_j E_{ij}E_{ji}=E_i \left(\f E2+N-2\right),
\qquad
\sum_{i,j} E_{ij}E_{ji}=E \left(\f E2+N-2\right).
\label{uncasimir}
\ee
Applying this constraint to the highest weight vector of a $\U(N)$-irrep, $v$ such that $E_i.v= l_i\,v$ with $l_1\ge l_2 \ge l_3\ge..\ge 0$ and $E_{ij}.v=0$ for all $i<j$, we solved these equations to get the eigenvalues $[l_1,l_2,l_3..,l_N]=[l,l,0,..,0]$. This corresponds to Young tableaux with two horizontal lines of equal length.

This shows that that the Hilbert space of $N$-valent intertwiners at fixed total area $J$ is actually an irreducible representation of $\U(N)$ with highest weight vector $[l,l,0,..,0]$ with $l=J$. This highest vector corresponds to the bivalent intertwiner with $j_1=j_2=J/2$ and all other spins set to 0. In \cite{UN}, we have checked this by showing that the explicit and direct counting of $\SU(2)$ intertwiners matches the dimensions and characters of these $[J,J,0,..,0]$ irreps of $\U(N)$. To conclude this review, we give the dimension of these $\U(N)$ representations, i.e the dimension of the intertwiner space with $N$ legs and total area $J$, in term of binomial coefficients:
\be
D_{N,J}\,\equiv\,
\dim \cH^{(J)}_N
\,=\,
\f{1}{J+1}\bin{J}{N+J-1}\bin{J}{N+J-2}
\,=\,
\frac{(N+J-1)!(N+J-2)!}{J!(J+1)! (N-1)!(N-2)!}\,.
\ee

\subsection{Creation and Annihilation Operators for Intertwiners}


We now introduce new operators which are quadratic in the annihilation operators $a,b$:
\be
\forall i\ne j, \, F_{ij}\equiv a_{i}b_{j}-a_{j}b_{i}, \quad F_{ij}=-F_{ji}.
\ee
The operator $F_{ij}$ decreases both  spins $j_i$ and $j_j$ by a half-spin.
The key property of these operators is that they are also invariant under global $\SU(2)$ transformations despite the fact that they do not contain creation operators:
\be
\left[\vJ\,,\,F_{ij}\right]=0.
\ee
In contrast with the $E$ operators, these new operators $F$ do not commute with the total area $E$: the $F$ decreases the total area by $-1$ while a $F^\dag$ operator increases it by $+1$.
Together, all the operators $E,F,F^\dag$ satisfy the following closed algebra~\footnotemark:
\bea
{[}E_{ij},E_{kl}]&=&
\delta_{jk}E_{il}-\delta_{il}E_{kj}\nn\\
{[}E_{ij},F_{kl}] &=& \delta_{il}F_{jk}-\delta_{ik}F_{jl},\qquad
{[}E_{ij},F_{kl}^{\dagger}] = \delta_{jk}F_{il}^{\dagger}-\delta_{jl}F_{ik}^{\dagger}, \label{EF}\\
{[} F_{ij},F^{\dagger}_{kl}] &=& \delta_{ik}E_{lj}-\delta_{il}E_{kj} -\delta_{jk}E_{li}+\delta_{jl}E_{ki}
+2(\delta_{ik}\delta_{jl}-\delta_{il}\delta_{jk}), \nn\\
{[} F_{ij},F_{kl}] &=& 0,\qquad {[} F_{ij}^{\dagger},F_{kl}^{\dagger}] =0.\nn
\eea
\footnotetext{
The constant shift  $2(\delta_{ik}\delta_{jl}-\delta_{il}\delta_{jk})$ in the $[F,F^\dag]$ commutator disappears if we use the alternative definition of the $\u(N)$ generators $\tilde{E}$. The price to pay is that the diagonal operators $\tilde{E}_i$ do not vanish on the vacuum state $|0)$.
}
We can interpret this new set of $F,F^\dag$ operators  respectively as annihilation and creation operators providing the space of $N$-valent intertwiners $\cH_N\,=\, \bigoplus_{J\in\N}\cH_N^{(J)}$ with a Fock space structure. The operators $E_{ij}$ allows to go from one state to another within each subspace $\cH_N^{(J)}$ while the operators $F_{ij},F^\dag_{ij}$ allows to move between subspaces with different total area $J$.

Moreover, we can express the scalar product operators once again in term of the $F$'s:
\be
\forall i\ne j,\,\vJ_{i}\cdot \vJ_{j}
\,=\,
-\f12F^\dag_{ij} F_{ij} +\f14E_iE_j.
\ee
One advantage of this formula over the expression of $\vJ_{i}\cdot \vJ_{j}$ in terms of the $E$'s is that this formula is explicitly symmetric in $i\leftrightarrow j$. It also provides an explicit relation between $F^\dag_{ij} F_{ij}$ and $E^\dag_{ij} E_{ij}$, which are equal up to terms in $E_i,E_j$.

In the following, it will be convenient to write the $E,F$ algebra in terms of the operators
\be
E_{\alpha}\equiv \sum_{ij}\alpha^{ij}E_{ij},\qquad F_{\Bz}\equiv \frac12 \sum_{ij} \overline{\Bz}^{ji}F_{ij},
\quad\textrm{with}\quad
(E_\alpha)^\dag =E_{\alpha^\dag},\quad F_{\Bz^{t}}=-F_\Bz.
\ee
We assume that $\Bz$ is an antisymmetric matrix and  that $\alpha$ is anti-hermitian.
The algebra then reads:
\bea
[E_{\alpha},E_{\beta}] = E_{[\alpha,\beta]}, &\qquad&{[}F_{\Bz},F_{\Bw}^{\dagger}] = E_{\Bw\Bz^{\dagger} }+\,\left(\tr \Bw\Bz^\dag\right)\,\id,\\
{[}E_{\alpha},F_\Bz] =  F_{(\alpha \Bz + \Bz \alpha^{t})}&\qquad& {[}E_{\alpha},F_{\Bz}^{\dagger}] =   F_{(\alpha \Bz + \Bz \alpha^{t})}^{\dagger}.\nn
\eea
For a hermitian matrix $\rho$, we can also write the commutation relations, which are the same up to a sign in the $[E,F^\dag]$ commutator:
$$
{[}E_{\rho},F_\Bz] = - F_{(\rho \Bz + \Bz \rho^{t})},
\qquad
{[}E_{\rho},F_{\Bz}^{\dagger}] =  + F_{(\rho \Bz + \Bz \rho^{t})}^{\dagger}.
$$

\section{Defining $\U(N)$ Coherent States}

\subsection{Working with Spinors}

We denote  $|z\ra \in \C^{2}$ a un-normalized spinor and by $\la z|$ its conjugate.
Explicitly:
$$
|z\ra = \begin{pmatrix}\zo\\\zi\end{pmatrix}\qquad\textrm{and}\qquad
\la z|=(\,\bzo,\bzi\,).
$$
We also introduce the $\SU(2)$ structure map  $\varsigma$:
\be
\varsigma\begin{pmatrix}\zo\\ \zi\end{pmatrix}
\,=\,
\begin{pmatrix}-\bzi\\ \bzo \end{pmatrix},
\qquad \varsigma^{2}=-1.
\ee
This is an anti-unitary map, $\la \varsigma z| \varsigma w\ra= \la w| z\ra=\overline{\la z| w\ra}$, and we will write the related state as
$$
|z]\equiv \varsigma  | z\ra.
$$
The Hermitian inner product and antisymmetric bilinear form are
$$
\la z|w\ra = \bzo\wo + \bzi\wi,\qquad \quad[z|w\ra  = \zo\wi-\zi\wo.
$$
We have the following relations:
\be
[z|w]= {\la w| z\ra},
\qquad [z|w\ra=- [w|z\ra, \qquad \overline{[z|w\ra}= \la w|z] .
\ee

A spinor defines a three-dimensional vector $\vec{J}(z)$ vector by projecting it over the Pauli matrices as follows:
\be
|z\ra \la z| = \f12 \left( {\la z|z\ra}\id  + \vec{J}(z)\cdot\vec{\sigma}\right),\qquad
|z][  z| = \f12 \left({\la z|z\ra}\id - \vec{J}(z)\cdot\vec{\sigma}\right),
\ee
where the Pauli matrices $\sigma_a$ are taken Hermitian and satisfying the normalization $\sigma_a^2=\id$.
The norm of the vector is $|\vec{J}(z)| = \la z|z\ra= |\zo|^2+|\zi|^2$ and its components are explicitly:
\be
J^x=2\,\Re\,(\zo\bzi),\qquad
J^y=2\,\Im\,(\zo\bzi),\qquad
J^z=|\zo|^2-|\zi|^2.
\ee
Notice that changing $\vec{J}(z)$ into $-\vec{J}(z)$ is achieved by mapping $z$ to $\sigma\,z$, or equivalently $|z\ra$ to $|z]$.
Let us also underline the fact that the spinor $z$ is entirely determined by the corresponding 3-vector $\vec{J}(z)$ plus the relative phase between $\zo$ and $\zi$. When looking at the geometric interpretation of coherent states and the $\U(N)$ action on them, we will see that this phase will play a non-trivial role.

\medskip

To construct coherent states, we associate a spinor $z_i$ to each leg of the intertwiner. Thinking of each leg as dual to an elementary face in the dual boundary surface, the vector $\vec{J}(z_i)$ plays the role of the normal vector to the $i$-th face. The norm $|\vec{J}(z_i)|$ then becomes the area of the corresponding face.
One important remark concerns the  closure constraint
 $\sum_{i} J(z_{i})=0$. Written in terms of spinors, it reads
\be
\label{closure1}
\sum_{i} |z_{i}\ra\la z_{i}|  =  A(z)\id,\qquad A(z) \equiv \frac12 \sum_{i} \la z_{i}|z_{i}\ra,
\ee
or equivalently in term of the spinor components:
\be
\label{closure2}
\sum_i \zo_i\,\bzi_i=0,\quad
\sum_i \left|\zo_i\right|^2=\sum_i \left|\zi_i\right|^2=A(z).
\ee
We interpret $A(z)$  as (half) the total area of the dual boundary surface.
In the following, we will often assumes that the  $N$ spinors $|z_{i}\ra $ satisfy this closure constraint.

Note that  there is a natural $\mathrm{SL}(N,\C)$-action of the set of $N$ spinors:
\be\label{SLaction}
|(u z)_{i}\ra  \equiv \sum_{j} u_{ij} |z_{j}\ra.
\ee
If $u$ is in $\U(N)$ this action preserves the closure relation:
\be
\sum_{i} |(u z)_{i}\ra \la (u z)_{i} | =
\sum_{i,j,k} u_{ij} \bar{u}_{ik}  |z_{j}\ra \la z_{k}| =
\sum_{j,k} (u^{\dagger}u)_{kj}  |z_{j}\ra \la z_{k}|
= \sum_{i} |z_{i}\ra \la z_{i}|.
\ee
There is also a natural diagonal  action of $\mathrm{SL}(2,\C)$ on the space of N spinors given by $|z_{i}\ra \to \lambda |z_{i}\ra$.

\subsection{Coherent States and the Geometric action of $\U(N)$ and $\mathrm{SL}(N,\C)$}


Given the $N$ spinors, we construct the following antisymmetric  matrix which is invariant under the diagonal  $\mathrm{SL}(2,\C)$ action:
\be
\Bz_{ij} \equiv  [ z_{i}|z_{j} \ra
\,=\,
{\zo_i\zi_j-\zo_j\zi_i},
\ee
and consider the corresponding creation operator:
$$
F^{\dag}_\Bz\,=\,
\f12\sum \Bz_{ji} F^{\dag}_{ij}
\,=\,
\f12\sum {\zo_j\zi_i-\zo_i\zi_j} F^\dag_{ij}
$$
Note that the matrix $\Bz_{ij}$ and the operator $F^{\dag}_\Bz$ are both holomorphic in  $z_{i}$ and $z_{j}$.

\begin{defi}
Given an integer $J$, we define the following coherent states
\be
|J,z_{i}) \equiv\frac{1}{ \sqrt{(J+1)}} \, \frac{(F^{\dagger}_{\Bz})^{J}}{J!} |0).
\ee
where $|0)$ is the Fock vaccua.

They are coherent with respect to the $\mathrm{SL}(N,\C)$ action:
\be
\label{unaction}
\forall u\in\mathrm{SL}(N,\C),\quad
\hat{u}\, |J,z_{i}) = | J, (u\,z)_{i}),
\ee
for the  transformation $u = e^{\alpha},\,\hat{u} = \exp(  E_{\alpha})$ parameterized in term of the  matrix $\alpha$.
These states are also covariant   under the diagonal $\mathrm{GL}(2,\C)$ action:
\be\label{cov}
\forall \lambda \in\mathrm{GL}(2,\C),\quad
|J, \lambda z_{i}) =  \left(\mathrm{det}(\lambda) \right)^{J}|J,z_{i}).
\ee
This imply that these states  are intertwiner states with total area $J$.
In particular this means that the states are invariant under the diagonal $\mathrm{SL}(2,\C)$ action.
\end{defi}

\begin{proof}
To start with, the vacuum state $|0)$ is obviously $\mathrm{GL}(2,\C)$-invariant. Moreover under a transformation $|z_{i}\ket \to \lambda |z_{i}\ket$ the antisymmetric bracket transforms as $ \Bz_{ij}\to \mathrm{det}(\lambda) \Bz_{ij}$ and so does
$F^{\dagger}_\Bz \to \mathrm{det}(\lambda) F^{\dagger}_\Bz$. This implies the covariance property (\ref{cov}).
Note that in particular this property implies that  the resulting state $|J,z_{i})$ is  $\SU(2)$-invariant and thus an intertwiner. Using the commutation relation $[E,F^{\dagger}_\Bz]=\,2\,F^{\dagger}_\Bz$ where $E=\sum_i E_i$ is the operator defining twice the total area, we easily check that the area is given by the integer $J$~:
$$
E\,|J,z_{i}) = 2J\,|J,z_{i}) .
$$
Let us now consider a  transformation $u=e^{\alpha}$ parameterized by an arbitrary  matrix $\alpha$. The action $z\arr \,u\,z$ leads to the following transformations of the matrix $\Bz$:
\be
\Bz \arr {u} \, \Bz \, u^t.
\ee
We can now compute the action of the  transformation $e^{E_\alpha}$ on the proposed coherent states:
$$
e^{E_\alpha}\,(F^{\dagger}_{\Bz})^{J} |0)
\,=\,
e^{E_\alpha}\,(F^{\dagger}_{\Bz})^{J}e^{-E_\alpha} |0)
\,=\,
\left(e^{E_\alpha}\,F^{\dagger}_{\Bz}\,e^{-E_\alpha}\right)^{J} |0),
$$
since $E_\alpha\,|0)=0$.
Then considering the commutation relation $[E_\alpha, F^\dag_\Bz] =  F^\dag_{\alpha\Bz+\Bz\alpha^{t}}$
we can easily derive the following identity:
\be
e^{E_\alpha} F^\dag_\Bz e^{-E_\alpha} = F^\dag_{e^{\alpha} \Bz e^{\alpha^{t}} }.
\ee
Note that a similar  relation holds for the $F$-operators:
$
e^{E_\alpha} F_\Bz e^{-E_\alpha} = F_{e^{-\alpha^{\dag}} \Bz e^{\bar{\alpha}} }.
\label{UNF}
$
This transformation is the same as $F^{\dag}$ if $\alpha$ is antiunitary, hence if $u\in U(N)$.

We immediately see that this fits the behavior given above of the matrix $\Bz$ under the action of $u=e^{\alpha}$, thus showing that our coherent states behave consistently under $\mathrm{SL}(N,\C)$ transformations:
$$
\hat{u}\, |J,z_{i}) = | J, (u\,z)_{i}).
$$

\end{proof}

\subsection{Coherent states norm}
In this section we assume that the closure identity $\sum_{i}|z_{i}\ra\la z_{i}| = A(z) \id $ is satisfied unless explicitly stated otherwise.
We introduce the following hermitian matrix
\be
\qquad \rho_{ij} \equiv \la z_{i}|z_{j} \ra.
\ee

Due to the closure constraint this  matrix together with $\Bz$ satisfy the relations
\be
\rho \rho = A\, \rho,\qquad
\rho^{t}\, \Bz =  \Bz\,\rho =  A\, \Bz,\qquad
\Bz^{\dagger}\Bz= A(z) \rho,\qquad
\Bz\Bz^{\dagger}= A(z) \rho^t= A(z) \bar{\rho}\,.
\ee
Therefore the generators $E_{\bar{\rho}},F_{z}, F^{\dagger}_{z}$ satisfy a $\sl_2$ algebra~\footnotemark:
\bea
[F_{\Bz},F_{\Bz}^{\dagger}] &=& A\, (E_{\bar{\rho}}+\,2A),\\
{[}E_{\bar{\rho}}, F_{\Bz}^{\dagger}] &=&  2 A\, F_{\Bz}^{\dagger},
\qquad
[E_{\bar{\rho}}, F_{\Bz}] =  -2 A\, F_{\Bz}.
\eea
\footnotetext{
Actually, the standard $\sl_2$ algebra is obtained by defining the shifted and renormalized generators:
$$
H = \f1{2A}(E_{\bar{\rho}} + 2 A), \quad X^+ = \f1A F^\dagger_z,  \quad X^-= \f1A F_z,
$$
which satisfy the usual commutation relations $[H,X^\pm]=\pm X^\pm$ and $[X_+,X_-]=2H$.
}
Notice the constant shift in $E_{\bar{\rho}}$ which is due to the term $\tr\Bz^\dag\Bz=A\,\tr\bar{\rho}=2A^2$ since $A=\bar{A}$ is real (and positive).
\begin{prop}
The norm of  states satisfying the closure constraint is  given by
\be
(J,z_{i}|J,z_{i}) =  (A(z))^{2J}=\left( \frac12 \sum_{i}\la z_{i}|z_{i}\ra\right)^{2J}.
\ee
\end{prop}

\begin{proof}

First, we iterate the commutation relation $[E_{\bar{\rho}},F^\dag_\Bz]=\,2A\,F^\dag_\Bz$ to get:
\be
E_{\bar{\rho}}(F_\Bz^\dag)^{n}
\,=\,
(F_\Bz^\dag)^{n}(E_{\bar{\rho}} +2nA).
\label{ErhoF}
\ee
Then using the commutator $[F_\Bz,F^\dag_\Bz]=\,A\,(E_{\bar{\rho}}+2A)$ allows us to compute:
\bea
{[}F_\Bz,(F_\Bz^\dag)^{n}]
&=&
A \sum_{k=0}^{n-1}(F_\Bz^\dag)^{n-k-1}(E_{\bar{\rho}}+2A)(F_\Bz^\dag)^{k} \nn\\
&=&
A(F_\Bz^\dag)^{n-1} \sum_{k=0}^{n-1}(E_{\bar{\rho}}+2(k+1)A)
\,=\,
An \,(F_\Bz^\dag)^{n-1} \left( E_{\bar{\rho}}+(n+1)A\right).
\eea
Since both $E$ and $F$ operators vanish on the vacuum state, $E_{\bar{\rho}}|0)=\,F_\Bz|0)=\,0$, we obtain:
\be
F_\Bz (F_\Bz^\dag)^{n}|0)
\,=\,
A^{2}n(n+1)  (F_\Bz^\dag)^{n-1}|0).
\label{FFd}
\ee
Thus iterating this relation, we finally compute:
\be
F_\Bz^{J} (F_\Bz^\dag)^{J}|0)= A^{2J}J!(J+1)!  |0).
\ee

\end{proof}

A side-product of this proof is that the commutation relation \Ref{ErhoF} implies that the coherent state $|J,z_i\ra$ is an eigenstate of the $\u(N)$ generator $E_{\bar{\rho}}$~:
\be
E_{\bar{\rho}}\,|J,z_i\ra\,=\, 2JA(z)\,|J,z_i\ra.
\ee

Now we would like to relax the  closure condition. We define the following 2 by 2 matrix
$$
X(z)\equiv \sum_{i} |z_{i}\ra \la z_{i}|.
$$
This is a positive hermitian matrix; therefore there exists a matrix $\lambda \in \mathrm{SL}(2,\C)$ such that $ X(z) = \sqrt{\det(X)} \lambda \lambda^{\dagger}$. The transformation $ \lambda$ is uniquely determined by this condition up to $\mathrm{SU}(2)$ transformations
$\lambda \to \lambda g$.
By construction the states $ |z'_{i}\ra = \lambda^{-1}|z_{i}\ra$ do satisfy the closure constraint and $A(z'_{i}) = \sqrt{\det(X)(z_{i})}$. This is an important point: we can always rotate by $\SL(2,\C)$ the $N$ spinors $z_i$ so that they satisfy the closer constraints. The present $\SL(2,\C)$-action is the same than the $\SL(2,\C)$-action used earlier in \cite{coh2,holomorph} to rotate the labels of coherent intertwiners in order that they satisfy the closure constraints.

Here, we use this construction and the $\mathrm{SL}(2,\C)$ invariance of the coherent states to compute the norm of an arbitrary state:
\be\label{normgen}
(J,z_{i}|J,z_{i}) =  (J,z'_{i}|J,z'_{i})
= A(z_{i}')^{2J}
=  \left(\det \left( \sum_{i} |z_{i}\ra \la z_{i}|\right)\right)^{J}.
\ee
This generalizes the computation of the norm for spinors which do not necessarily satisfy the closure constraints.

What is interesting about this formula is the fact that the states satisfying the closure constraints  maximize this norm for fixed area.
Indeed, suppose that  $|z_{i}\ra$ is a set of spinors of fixed area $A$:
$$
\sum_{i} \la z_{i}|z_{i}\ra = 2A(z)=  \tr X(z).
$$
Let us define the dispersion $ \Delta(z)^{2} \equiv \frac12 \tr \left(X(z)- A(z) \id \right)^{2}$. It measures how far we are from the closure constraints. Indeed we can easily express it in term of the 3-vectors:
\be
\Delta(z)^{2}
\,=\,
\f14\,\left|\sum_i \vJ(z_i)\right|^2.
\ee
Then using a standard identity on 2$\times$2 matrices relating their determinant to their trace, $(\tr X)^{2}-\tr(X^{2})= 2\det X$, we evaluate the norm of the coherent states in term of the area $A(z)$ and the dispersion $\Delta(z)$:
\be
\det(X(z))= A^{2}(z) - \Delta(z)^{2}.
\label{dispersion}
\ee
Thus
\be
\frac{(J,z_{i}|J,z_{i})}{A^{2J}(z_{i})} = \left(1-\left(\frac{\Delta(z)}{A(z)} \right)^{2}\right)^{J} < e^{-J\left(\frac{\Delta(z)}{A(z)} \right)^{2} }
\ee
This means that in the semi-classical limit $J\to \infty$ the contribution from states not satisfying the closure constraints are exponentially suppressed compare to the one satisfying the constraints. This is very similar to the earlier result on coherent intertwiners derived in \cite{coh1}: coherent intertwiners that did not satisfied the classical closure constraints were exponentially suppressed in the decomposition of the identity on the intertwiner space.

\subsection{Scalar products between $\U(N)$ Coherent States}
\label{alaperelomov}

\begin{prop}
The  scalar product between the coherent states exponentiate as a simple power of the area $J$~:\be
( J,z_{i}|J,w_{i}) = (z_{i}|w_{i})^{J} \label{expJ}
\ee
where
\be\label{prod}
(z_{i}|w_{i}) \equiv \mathrm{det}\left(\sum_{i}|w_{i}\ra\la z_{i}|\right)
\,=\,
\left(\f12\tr \,\Bz^{\dagger}\Bw\right)
\ee
\end{prop}
Note that the determinant appearing in the  scalar product is the determinant of the 2-dimensional matrix
 $X(z,w)\equiv \sum_{i}|w_{i}\ra\la z_{i}|$.
 This matrix is holomorphic in $w_{i}$ and antiholomorphic in $z_{i}$.
Thus the previous formula  is clearly an analytical continuation of the  formula \Ref{normgen} for the norm outside the diagonal case  $z_{i}=w_{i}$.
In itself this provides the justification for this scalar product since a biholomorphic function $F(\bar{z},w) $ is entirely determined by its value $F(\bar{z},z)$ on the diagonal.
We give nevertheless a  constructive proof now, which sheds new light on the nature of these coherent states:

\begin{proof}Let us start to show the equality between the two expressions in (\ref{prod}).
One can make use of the spinorial identity
\be
|z][w| +|w\ra\la z| =\la z|w \ra \,\id
\ee
in order to rewrite $\tr \,\Bz^{\dagger}\Bw=\sum_{i,j} \bar{\Bz}_{ij}{\Bw}_{ij}$ as
\bea
\tr \,\Bz^{\dagger}\Bw &=&
\sum_{i,j}  \la z_{j}|z_{i} ] [ w_{i}|w_{j}\ra= \sum_{i,j}\la w_{i}|z_{i}\ra  \la w_{j}|z_{j}\ra  -
\sum_{i,j}\la z_{j}|w_{i} \ra \la  z_{i}|w_{j}\ra\\
&=&    (\tr X)^{2} -\tr X^{2} = 2 \mathrm{det} X
\eea
where we have introduced the two-dimensional matrix $X= \sum_{i} |w_{i}\ra \la z_{i}| $.

In order to prove (\ref{expJ}), we proceed in two steps: we start by proving it in the case $J=1$.
In this case
\bea
 (z_{i}|w_{i})\equiv( 1,z_{i}|1,w_{i})
&=&
\frac12 ( 0|F_\Bz F^{\dagger}_\Bw|0)=\frac12 ( 0| [F_\Bz, F^{\dagger}_\Bw]|0) \nn\\
&=&
\frac12 ( 0| \,\left(E_{ \Bw\Bz^{\dagger}}+\tr\,\Bz^\dag\Bw\right)\,|0) = \frac12\tr\,  \Bz^\dag\Bw.
\eea
This  establishes the result for $J=1$.

To generalize this result to an arbitrary integer $J$, we show that the state $|J,z_i\ra$ is actually directly the $J$-th tensor power of the state $|J=1,z_i)$:
We first remark that for the following choice of spinors,
$$
\zeta_1=\bin{0}{1},\qquad
\zeta_2=\bin{1}{0},\qquad
\zeta_k=\bin{0}{0}\quad \textrm{for}\quad k>2,
$$
%
%
the vector $|J,\zeta_{i})$ is the normalized highest weight vector  of the $\U(N)$-representation  of weight $[J,J,0,\cdots,0]$.
Indeed it is straightforward to check that $F^\dag_{{\mathbb \zeta}}=F_{21}^\dag$ and using (\ref{EF}) that therefore
\be
E_{kl} |J, \zeta_{i}\ra =0 \quad \mathrm{for}\,\, k<l \qquad
\mathrm{and}\quad E_{kk} |J, \zeta_{i}\ra = J(\delta_{k1}+ \delta_{k2}) |J, \zeta_{i}).
\ee
Moreover, we know from the proposition in previous section that this state as {\it unit} norm.
Since the tensor product of two normalized highest-weight states is still a normalized highest-weight state we have
\be
{|J, \zeta_{i})}\otimes{|K, \zeta_{i})}= {|J+K, \zeta_{i})}, \quad \mathrm{thus}\quad |J, \zeta_{i})=  {|1, \zeta_{i})^{\otimes J}}
\ee
Now, an arbitrary coherent state $|J, z_{i})$,  where $|z_{i}\ra$ satisfies the closure constraints,
can be obtained by the action of a $\U(N)$ group element $u$ on such highest-weight vector. Indeed, we choose the group element such that:
\be
u_{i1} = \frac{\zo_i}{\sqrt{A(z)}},\quad u_{i2} = \frac{\zi_i}{\sqrt{A(z)}}.
\label{umatrix}
\ee
This is truly a unitary matrix since the two vectors $u_{i1}$ and $u_{i2}$ are orthonormal due to the closure condition $\Ref{closure2}$  imposed on the spinors. Then it is straightforward to check that $|z_i\ra=\sqrt{A}\,\sum_j u_{ij} \zeta_j $. Thus using the $\U(N)$ action \Ref{unaction} on the coherent states, we obtain:
\be
\frac{1}{ (A(z))^{J}}\,|J, z_{i}) =   u\,|J, \zeta_i)
=u\,{|1, \zeta_{i})^{\otimes J}}
=  \, \frac{1}{(A(z))^{J}}\,|1, z_{i})^{\otimes J}.
\ee
This allows us to conclude that:
\be
(J,z_{i}|J, w_{i}) = ( 1,z_{i}|1, w_{i})^{J} = \left(  \det\, X\right)^{J}.
\ee
when both $z_{i},w_{i}$ satisfies the closure condition.
Using the $\mathrm{SL}(2,\C)$ invariance of the coherent state we can relax the restriction that the closure condition is satisfied since
given an arbitrary set of $N$ spinors $|z'_{i}\ra$ we can always find a $\mathrm{SL}(2,\C)$ element $\lambda$ such that
$|z_{i}\ra \equiv \lambda |z'_{i}\ra$ satisfies the closure condition.
\end{proof}

The important side-product of this proof is that our coherent states that we constructed through the action of the $F^\dag$ operators are actually coherent states {\it \`a la} Perelomov, that is obtained trough the action of our group $\U(N)$ on highest weight vectors. What's particularly interesting is that the highest weight vectors correspond to bivalent intertwiners, with all spinors set to 0 but two of them. Then the $\U(N)$ action allows to go from these bivalent intertwiners to any semi-classical $N$-valent configurations.
More explicitly, given $N$ spinors $|z_{i}\ra$ such  that $\Bz_{12}\neq 0$, we define the following new set of spinors,
\be
|Z_{k}\ra\equiv  \frac{\Bz_{2k}}{\Bz_{21}} |0\ra - \frac{\Bz_{1k}}{\Bz_{21}} |0]\,,
\ee
where we have introduced the simple convention for spinors:
$$
|0\ra\,\equiv\,\matri{c}{1\\0},\qquad
|0]\,\equiv\,\matri{c}{0\\1}.
$$
This spinors are such that $|Z_{1}\ra =|0\ra$ and $Z_{2}=|0]$, and they are related to the original set of spinors by an $\mathrm{GL}(2,\C)$ transformation:
\be
|Z_{k}\ra =\lambda |z_{k}\ra,
\qquad
\lambda
\equiv\f{\left( |0\ra [z_{2}| -|0][z_{1}|\right)}{[z_{2}|z_{1}\ra}
\,=\,
\f1{[z_{1}|z_{2}\ra}\matri{cc}{z_2^1 & -z_2^0 \\ -z_1^1 & z_1^0},
\qquad
\det(\lambda)= \Bz_{12}{}^{-1}.
\ee
Therefore the new coherent state is related to the original by a simple factor, $|J,z_{i})= (\Bz_{12})^{J} |J,Z_{i})$.
We can now construct explicitly a $\mathrm{SL}(N,\C)$ group element $u$ mapping $\zeta_{i}$ onto $Z_{i}$ and find an explicit relationship between the coherent state and the highest weight state:
\be
|J,Z_{i}) = \exp\left( \sum_{k=3}^{N} Z_{k}^{0}E_{k1}+ Z_{k}^{1}E_{k2}\right) |J,\zeta_{i}).
\ee
Finally, since we are dealing with coherent states {\it \`a la} Perelomov, we also expect these $\U(N)$ coherent states to minimize the uncertainty relations, as we check below.

\subsection{Expectation Values and Uncertainties}

Using the $\mathrm{SL}(N,\C)$ action on the coherent states (\ref{SLaction}) and the scalar product between these states, we can now easily compute the expectation values of  $E$-operators on such states.  Indeed, considering the $\mathrm{SL}(N,\C)$ transformations $u_{\alpha}=\exp( E_{\alpha})$ where $E_{\alpha}=\sum_{i,j}\alpha^{ij}E_{ij}$ for arbitrary N by N matrix $\alpha$ we can write:
\be\label{Uab}
(J,z_{i}|\hat{u}_{\beta}^{\dag}\hat{u}_{\alpha} |J, z_{i})
\,=\,
(J,(u_{\beta}z)_{i}|J, (u_{\alpha}\,z)_{i}\ra
\,=\,
\left(\f12\tr\, (u_{\beta}\Bz u_{\beta}^{t})^\dag u_{\alpha}\Bz u^{t}_{\alpha} \right)^J.
\ee
We can now expand this expression in power series in power of $\alpha$ and $\beta$.
At leading order, we recover the norm of the coherent state:
$$
(J,z_{i} |J, z_{i})=\left(\f12\tr\, \Bz^\dag\Bz\right)^J=A^{2J}(z).
$$
where the last equality is valid for states solution of the closure constraints.
Then, at first order in $\alpha$, we obtain after a straightforward calculation:
\bea
\la E_{\alpha} \ra \equiv \frac{(J,z_{i}|E_{\alpha} |J, z_{i})}{(J,z_{i}|J, z_{i})}
= J \frac{(z_{i}|\delta_{\alpha}z_{i})}{(z_{i}|z_{i})}=
2J\frac{\tr\,\alpha\Bz\Bz^\dag}{\tr\, \Bz^\dag\Bz}
\eea
with $\delta_{\alpha} \Bz = \alpha \Bz +\Bz \alpha^{t}$.

It is enough to evaluate the expectation value on states  $|\hat{z}_{i}\ra$ that satisfy the closure constraint $\sum_{i}|\hat{z}_{i}\ra\la \hat{z}_{i}|= A(\hat{z})\id$ and
which are normalized $A(\hat{z}_{i})=1$. For such states $\Bz^{\dag}\Bz = \rho$ and
we get the expectation value:
\be
\la E_{ij}\ra = J\rho_{ij} =
J \,\la \hat{z}_i|\hat{z}_j\ra.
\ee
Note that if necessary we can rewrite this expression in terms of the un-normalized spinors using $|\hat{z_{i}}\ra =|z_{i}\ra / \sqrt{A(z)}$.
The natural question is the ``geometrical" meaning of this expectation value.
So considering the Gram matrix $\rho_{kl}=\la \hat{z}_k|\hat{z}_l\ra$, we know that $\rho^2=\rho$ for normalized states solution of the closure constraints, so that it has eigenvalues 0 or $1$. Moreover, we also know that $\tr\,\rho=2A(\hat{z})=2$, so that we know that the eigenvalues of $\rho$ are $[1,1,0,0,..,0]$. This means that the eigenvalues of the expectation value matrix $\la E_{kl}\ra$ are independent of the spinors $\hat{z}_i$ and equal to $[J,J,0,0,..,0]$, which fits perfectly with the highest weight of the $\U(N)$ representation.

Moreover, the two eigenvectors of $\rho$ are easy to identify. Indeed, using the closure condition on the spinors, we have:
\be
\sum_l\rho_{kl} \la \hat{z}_l|
\,=\,
\sum_l\la \hat{z}_k|\hat{z}_l\ra \la \hat{z}_l|
\,=\,
\la \hat{z}_k|.
\ee
Thus the two eigenvectors of the $\rho$-matrix are the two components of the spinors, $\zo_k$ and $\zi_k$.

If we now expand (\ref{Uab}) at second order we get the expression for the  spread of the coherent states
\be
\la E_{\beta^{\dagger}} E_{\alpha} \ra -\la E_{\beta^{\dag}}\ra \la E_{\alpha}\ra = J \left(
\frac{(\delta_{\beta} z_{i}|\delta_{\alpha} z_{i})}{(z_{i}|z_{i})} -\frac{(\delta_{\beta}z_{i}|z_{i})(z_{i}|\delta_{\alpha}z_{i})}{(z_{i}|z_{i})^{2}}\right)
\ee
with $\delta_{\alpha} \Bz = \alpha \Bz +\Bz \alpha^{t}$.

If one evaluate this expression for normalized states $|\hat{z}_{i}\ra$ satisfying the closure constraint we get
\be
\la E_{\beta^{\dagger}} E_{\alpha} \ra -\la E_{\beta^{\dag}}\ra \la E_{\alpha}\ra = J \left(
\,\tr\,\left(\rho^{t} \beta^{\dagger} \alpha+\Bz^\dag\alpha\Bz\bar{\beta}\right) - (\tr \rho^{t} \alpha)(\tr \rho^{t} \beta)
\right).
\ee
Evaluating this expression for $\beta_{ab}=\alpha_{ab} = \delta_{ab}^{ji}$ gives for instance the evaluation
\bea\label{offdiag1}
\Delta^{2}E_{ij}\equiv\la E_{ij} E_{ji} \ra -\la E_{ij}\ra \la E_{ji}\ra &=& J \left(\rho_{ii} - |\Bz_{ij}|^{2} -|\rho_{ij}|^{2}\right) \\
&=& J |J(\hat{z}_{i})|(1-|J(\hat{z}_{j})|)).
\eea
Here we have used the spinorial definitions of the matrices $\rho$ and $\Bz$ to express their matrix elements in term\footnote{
\bea
&&\rho_{kk}=\la z_k|z_k\ra=|J(z_k)|,\nn\\
&&\rho_{kl}\rho_{lk}=\la z_k|z_l\ra\la z_l|z_k\ra=
\f12\left(|J(z_k)||J(z_l)|+J(z_k)\cdot J(z_l)\right),\nn\\
&&|\Bz_{kl}|^2=\la z_k|z_l][ z_l|z_k\ra=
\f12\left(|J(z_k)||J(z_l)|-J(z_k)\cdot J(z_l)\right).\nn
\eea}
of the 3-vectors  $\vec{J}(z_k)$ and $\vec{J}(z_l)$.
In terms of these vectors the closure  condition reads $\sum_{i}\vec{J}(z_{i})=0$ and the normalization condition
$\sum_{i}|J(\hat{z}_{i})|=2A(\hat{z})=2$. We can also easily compute the spread for the diagonal operators $E_k=E_{kk}$. Applying the previous formula to the matrices $\alpha_{ab}=\delta_{a}^{i}\delta_{b}^{i}$ $\beta_{ab}=\delta_{a}^{j}\delta_{b}^{j}$, we derive the following uncertainties:
\bea
\la E_iE_j\ra-\la E_i\ra\la E_j\ra
={J}\left(\delta_{ij} \rho_{ii}+ |z_{ij}|^2-\rho_{ii}\rho_{jj}\right)
=
-\frac{J}{2}\,\left(|J(\hat{z}_i)||J(\hat{z}_j)|+{J(\hat{z}_i)\cdot J(\hat{z}_j)}\right) + J \delta_{ij}\,.\label{diag2}
\eea
A first remark is that this shows that these coherent states are peaked as the total area $J$ grows to infinity. Indeed, the uncertainties $\Delta E_{ij}$ grow as $\sqrt{J}$ while the mean values $\la E_{ij}\ra$ grow as ${J}$. Thus these $\U(N)$ coherent states are considered as semi-classical states in the asymptotic regime $J\arr\infty$ where $\Delta E_{ij}/\la E_{ij}\ra\sim 1/\sqrt{J} \arr 0$.

We can do two easy checks on this formula. Indeed, the coherent state $|J,z_i\ra$ is an eigenvector of both the total area $E=\sum_i E_{ii}$ and the generator $E_\rho$. Therefore, both these operators would have a vanishing spread. These two choices correspond respectively to $\alpha=\beta =\id$ and $\alpha=\beta=\rho$. An easy calculation shows that the expression above vanishes in both cases as expected since $\delta_{\id}\Bz=\delta_{\rho}\Bz=2\Bz$.

Also, using the quadratic constraint \Ref{uncasimir} between the $\U(N)$ Casimir operator and the total area $E$ and taking into account that the coherent states are eigenvectors of the total area with eigenvalue $E=2\,J$, we easily compute:
\be
\la \sum_{i,j} E_{ij}E_{ji}\ra
\,=\,
\la E\left(\f E 2 +N-2\right)\ra
\,=\,
2J\,(J+N-2).
\ee
Combining this with the expectation values of the $E_{ij}$ operators, we obtain the invariant uncertainty~:
\be
\Delta\,\equiv\,\la \sum_{i,j} E_{ij}E_{ji}\ra-\sum_{i,j}\la  E_{ij}\ra\la E_{ji}\ra
\,=\,
2J\,(J+N-2)  \,-\,  {J^2}\tr\,\rho^2
\,=\,
2J\,(N-2).
\ee
One can actually check this formula by directly summing up the uncertainties \Ref{diag2}-\Ref{offdiag1}. The conclusion is that the averaged standard deviation $\sqrt{\Delta}/N$ grows as $\sqrt{J/N}$ while the average mean values $\sqrt{\sum_{ij}\la E_{ij}\ra^{2}}/N$ grow linearly in $J/N$. So in the large N limit we still need the ratio $J/N$ to go to $\infty$ in order to reach the semi-classical limit.

These off-diagonal relations allow us to compute the expectation values of the scalar product operators $\vec{J}_k\cdot\vec{J}_l$. Indeed, we recall that for $k\ne l$, these are easily expressed in term of the $\u(N)$ generators~:
\bea
4\,\vec{J}_k\cdot\vec{J}_l&=& 2\,E_{kl}E_{lk} - E_kE_l -2E_k\nn\\
&=&E_{kl}E_{lk}+E_{lk}E_{kl} -E_kE_l -(E_k+E_l).\nn
\eea
Applying the previous results give us:
\bea
4\la \vec{J}_k\cdot\vec{J}_l\ra &=&
J^2\,{\vec{J}(\hat{z}_k)\cdot\vec{J}(\hat{z}_l)}
+\frac{J}{2}\,\left({\vec{J}(\hat{z}_k)\cdot\vec{J}(\hat{z}_l)} -3
{|\vec{J}(\hat{z}_k)|\,|\vec{J}(\hat{z}_l)|}
 \right),\\
&\underset{J\gg1}{\sim}&
{J^2}\,\vec{J}(\hat{z}_k)\cdot\vec{J}(\hat{z}_l).
\eea
This confirms the precise geometrical interpretation of the 3-vector $\vec{J}(z_i)$ as the normal vector to the face dual to the $i$th leg of the intertwiner in the semi-classical limit $J\arr\infty$.
We conclude this section with the calculation of the $\U(N)$-invariant uncertainty. Thus our $\U(N)$ coherent states are semi-classical intertwiner states, as we expected since we have already shown that they are coherent states {\it \`a la} Perelomov.


\subsection{The $F$-action on Coherent States}

After having studied the action of the $\u(N)$-operators $E_{ij}$ on the coherent states, we can investigate the action of the annihilation operators. Let us have a closer look at $F_\Bw\,|J,z_i)$ for arbitrary spinors $z_i$ and antisymmetric matrix $\Bw$.

\begin{prop}
The action of the $F$ operator on the coherent states is given by
\be
F_{\Bw} |J,z_{i}) = \sqrt{J(J+1)}  (w_{i}|z_{i}) \,\,|J-1,z_{i})
\ee
where
$
(z_{i}|w_{i}) =\left(\f12\tr \,\Bz^{\dagger}\Bw\right)
$
is the elementary scalar product.
\end{prop}
\begin{proof}
To show this proposition, we introduce the unitary matrix $u$ which maps the bivalent ansatz $\zeta_{i}$ to $z_i$ as introduced earlier in the section \ref{alaperelomov}. Assuming that $z_i=\sqrt{A}\,(u\zeta)_i$, we have~:
\be
F_\Bw\,|J,z_i)
\,=\,
A(z)^{J}\,F_\Bw \,u\,|J,\zeta_i)
\,=\,
A(z)^{J}u\,(u^{-1}\,F_\Bw \,u)\,|J,\zeta_i)
\,=\,
A(z)^{J}u\,F_{u^{\dag}\Bw \bar{u}} \,|J,\zeta_i)\,,
\ee
if we use the $\U(N)$ action on the $F$-operators derived earlier in eqn.\Ref{UNF}.

The particularity of the highest weight vector $|J,\zeta_i\ra$ is that the only $F$-operator which doesn't vanish on it is $F_{12}$. Indeed, we have by definition:
\be
|J,\zeta_i)\,
\,=\,
\f{1}{J!\,\sqrt{J+1}}
\, (F^\dag_{12})^J \,|0).
\ee
Thus, we can easily compute the action of the operator $F_{u^{t}\Bw u}$ on that state:
\be
F_{u^{\dag}\Bw \bar{u}} \,|J,\zeta_i\ra=
\,\f{(u^{t}\Bw^{\dag} u)_{12}}{J!\,\sqrt{J+1}}\, F_{12}(F^\dag_{21})^J \,|0)
=\f{(u^{t}\Bw^{\dag} u)_{12}}{J!\,\sqrt{J+1}}\, J(J+1)\,(F^\dag_{12})^{J-1} \,|0),
\ee
using the formula \Ref{FFd} giving in the commutation relation between $F$ and $(F^\dag)^J$. Then putting the pieces together, we obtain:
\be
F_\Bw\,|J,z_i\ra
\,=\,
A^J\,(u^{t}\Bw u)_{21}\,\sqrt{J(J+1)}\,u\,|J-1,\zeta_i\ra
\,=\,
A\,(u^{t}\Bw u)_{21}\,\sqrt{J(J+1)}\,|J-1,z_i\ra.
\ee
Finally, we can evaluate the pre-factor explicitly by using the formula for the matrix elements for the unitary matrix \Ref{umatrix}:
\be
A\,(u^{t}\Bw^{\dag} u)_{21}=
\sum_{i,j}\zi_i \Bw^{\dag}_{ij} \zo_j
=\f12\tr\,\Bw^\dag\Bz.
\ee
\end{proof}
To conclude, the action of the operators $F_\Bw$ on the $\U(N)$ coherent states is very simple. They truly act like annihilation operators, without acting on the spinor labels $z_i$ but simply on the total area $J$ by decreasing it by one quantum.
The pre-factor $\sqrt{J(J+1)}$ is nevertheless rather intriguing.

\subsection{On the Grassmanian}

Before studying more properties of  our coherent states we would like to investigate further
the space of labels of these coherent states.
Let us start by recalling that our $\U(N)$ coherent states are defined up to an overall transformation by
$\mathrm{GL}(2,\C)$: $z_i^{A}\arr\lambda^{A}_{B}\,z_i^{B}$ or $ |z_{i}\ket \to \lambda |z_{i}\ra$. As we have seen, this transformation  amounts to a simple rescaling of the corresponding states $|J,z_i)$~:
\be
\Bz\arr \mathrm{det}(\bar{\lambda})\,\Bz,\quad
F_\Bz^\dag\arr \mathrm{det}({\lambda}) F_\Bz^\dag,\quad
|J,\lambda\,z_i\ra\,=\,\mathrm{det}({\lambda})^{J} |J,z_i).
\ee
Thus our coherent states are {\it holomorphic} functions on $\C^{2N}$, moreover they are  homogeneous function of degree $J$ under $\mathrm{GL}(2,\C)$ transformations.

Therefore, the rays of $\U(N)$ coherent states are actually labeled by N two-planes in $\C^{N}$  or equivalently elements
of the Grassmanian space $\mathrm{Gr}_{2,N}$.
\be
\mathrm{Gr}_{2,N} \equiv \C^{2N}/\mathrm{GL}(2,\C) = \C\mathbb{P}^{2N-1}/\mathrm{SL}(2,\C).
\ee
A more elaborate way to describe coherent states is first to consider a line bundle over the Grassmanian space $ L_{2,N} \equiv
\C^{2N}/\mathrm{SL}(2,\C) $, so that $ G_{2,N} =L_{2,N}/\C$. Taking tensor product $L_{2,N}^{\otimes J}$ of this fundamental line bundle we get all possible determinant line bundle over $ G_{2,N}$. The coherent states are then {\it holomorphic sections } of $L_{2,N}^{\otimes J}$ (i-e holomorphic functions on $L_{2,N}$ homogeneous of degree $J$).

In order to understand better the geometry of $\mathrm{Gr}_{2,N}$ let us remark that given N spinors $|z_{i}\ket \in \C^{2}$ we can construct the 2 by 2 matrix $ X= \sum_{i} |z_{i}\ket\bra z_{i}|$ which transform under the
$GL(2,\C)$ action by $ X \to \lambda X \lambda^{\dagger}$.
Using this action we can always impose that $X=\id$, the residual invariance being then just a $\mathrm{U}(2)$ invariance.
Therefore if we define $ H(z_{i}) \equiv  \sum_{i} |z_{i}\ket\bra z_{i}| -\id $, we can equivalently characterize the Grassmanian space
as the quotient space
\be
\mathrm{Gr}_{2,N} = H^{-1}(0)/ \mathrm{U}(2).
\ee
This expression expresses that $\mathrm{Gr}_{2,N}$ is a symplectic quotient, hence a symplectic manifold
since $H$ is an hamiltonian generating the $\U(2)$ action on $\C^{2N}$, once we equipped $\C^{2N}$ with the canonical symplectic structure $\{\bar{z}_{i}^{A},z_{j}^{B}\} = \delta_{ij} \delta^{AB}$.
Imposing $H=0$ amounts to impose the closure condition $\sum_{i}\vec{J}(z_{i})=0$ and the normalization condition $\sum_{i}|\vJ(z_{i})|=2$
when written in terms of the 3-vectors.
In fact since it is a symplectic quotient of a Kahler manifold this is a K\"ahler manifold itself with K\"ahler potential proportional to
$K(z_{i})= \ln\left( z_{i}|z_{i} \right)$.
The metric is therefore given by $g = \partial \bar{\partial} K$ or
\be
\rd s^{2} = \frac{ (\rd z_{i}| \rd z_{i})}{(z_{i}|z_{i})} - \frac{ (\rd z_{i}| z_{i})}{(z_{i}|z_{i})}\frac{ (z_{i}| \rd z_{i})}{(z_{i}|z_{i})}.
\ee
and we recognize the metric arising in the computation of the dispersion relation \Ref{dispersion}.

Finally as seen in the previous sections, there is a natural action of  $\mathrm{U}(N)$ on  $\mathrm{Gr}_{2,N}$ given by
\be
|z_{i}\ket \to \sum_{j}u_{ij} |z_{j}\ket
\ee
which preserves the constraint $ H(u z)=H(z)$.
Using this action we can map any element of $H^{-1}(0)$ onto the canonical element
$$|\zeta_{1}\ket=|0\ket,\quad |\zeta_{2}\ket =|0],\quad |\zeta_{k}\ket=0 \quad \textrm{for}\quad k\ge 3 $$
so that $|z_{i}\ket = \sum_{j}u_{ij} |\zeta_{j}\ket$. This shows that  $\mathrm{Gr}_{2,N}$ is a $\U(N)$-homogeneous manifold
 quotient of $\mathrm{U}(N)$ by the isotropy group of $|\zeta_{i}\ket$ which is given by $ \mathrm{U}(N-2)\times \mathrm{U}(2)$. Therefore
 \be
 \mathrm{Gr}_{2,N} = \frac{\U(N)}{\U(N-2)\times \U(2)}.
 \ee
 The mapping  $|z_{i}\ket \to \Bz_{ij}= \bra z_{i}|z_{j}]$ defines the Pl\"ucker embedding  of $\mathrm{Gr}_{2,N}$ into
 the projective space $\mathbb{P}(\C^{N(N-1)/2})= \mathbb{P}(\C^{N}\wedge \C^{N})$.
 The image of the Pl\"ucker embedding is characterized by the following quadratic relations
 \be
 \Bz_{ij}\Bz_{ab} = \Bz_{ia}\Bz_{jb}-\Bz_{ja}\Bz_{ib}.
 \ee
 We can therefore describe the Grassmanian as the subset of $\mathbb{P}(\C^{N}\wedge \C^{N})$ modulo the Pl\"ucker relations.
Note that the fundamental coherent state product $(w_{i}|z_{i}) =1/2\tr \Bw^{\dagger} \Bz$ is the natural hermitian product
of $\mathbb{P}(\C^{N}\wedge \C^{N})$.

\medskip

From this analysis one clearly sees that the space of  states (and not only rays) of arbitrary spin $J$ is isomorphic (as a vector space) to the space of holomorphic function on  $L_{2,N}$ which is a line bundle over the Grassmanian:
\be
\cH_N=
\bigoplus_J \cH_N^{(J)}
\,=\,
\mathrm{Hol}\left(L_{2,N}\right),
\ee
where $S(G)$ denote the subgroup of $G$ of determinant one matrices.

In the previous work \cite{UN}, we had shown that the space $L_{2,N}/\R$  could be interpreted as the space of polyhedra with $N$ faces and fixed total area:
\be
L_{2,N}/\R\equiv \frac{\U(N)}{S(\U(N-2)\times \U(2))}
\sim P_{N}(A)\times \U(1)^{N},
\ee
where $S(G)$ denote the subgroup of $G$ of determinant one matrices.
Here $P_{N}(A)$ is the space of convex polyhedra with fixed total area $A$, only trivalent vertices and possessing $N$ faces (and thus $3(N-2)$ edges and $2(N-2)$ vertices). Then we have an extra $\U(1)$-phase for each face of the polyhedron.

The present $\U(N)$ coherent state construction provide this polyhedron picture with a concrete geometrical realization. Indeed, it gives coherent semi-classical states labeled by the normals $\vec{J}(z_i)$ to the $N$ faces with an explicit action of $\U(N)$ on them. Then one can reconstruct the classical polyhedron from these normals $\vec{J}(z_i)$. Actually, we can go beyond this since the spinors $z_i$ contain information about the normal $\vec{J}(z_i)$ and an extra phase, which we naturally identify as the $\U(1)$-phase appearing in the previous isomorphism between the Grassmannian space and the space of polyhedra.

We propose to interpret this $\U(1)$-phase as a choice of 2d-frame (i.e orthonormal basis) for each face. Its precise geometrical role can be investigated by analyzing the $\U(N)$ action on the new coherent states that we have defined. Such an interpretation would be consistent with the twisted geometries describing the discrete geometry of spin network states in loop quantum gravity \cite{twisted}.

To be more explicit, let us first give the expression of a spinor $z$ in term of its associated normal $\vec{J}(z)$. Let us recall that this normal is simply the projection of the matrix $|z\ra\la z|$ on the Pauli matrices. Then we can write:
\be
\vec{J}(z)=\tr\,\vec{\sigma}\,|z\ra\la z|
\qquad\Rightarrow\quad
\zo=e^{i\phi}\,\sqrt{\f{|\vec{J}|+J^z}{2}},\quad
\zi=e^{i(\phi-\theta)}\,\sqrt{\f{|\vec{J}|-J^z}{2}},\quad
\tan\theta=\f{J^y}{J^x},
\ee
where $e^{i\phi}$ is an arbitrary phase.

Now looking at the $\U(N)$ action on the $N$ spinors , $z_i\arr\,\sum_j u_{ij}z_j\equiv\tz_i$, it acts non-trivially on the $N$ normals:
\be
\vec{J}(z_i)=\tr\,\vec{\sigma}\,|z_i\ra\la z_i|
\,\arr\,
\vec{J}(\tz_i)=\sum_{j,k} u_{ij}\bar{u}_{ik}\,\tr\,\vec{\sigma}\,|z_j\ra\la z_k|.
\ee
The important point is that the transformed normals $\vec{J}(\tz_i)$ do not only depend on the initial normals $\vec{J}(z_i)$. They will also depend on the phases of the spinors.
Indeed, let us change the phases of the spinors $z_i$, this does not change the normals:
\be
\vec{J}(e^{i\phi_i}z_i)=\vec{J}(z_i).
\ee
Nevertheless, the normals after $\U(N)$ transformation will not match:
\be
\vec{J}((u\,e^{i\phi}z)_i)
=
\sum_{j,k} e^{i(\phi_j-\phi_k)}u_{ij}\bar{u}_{ik}\,\tr\,\vec{\sigma}\,|z_j\ra\la z_k|
\ne
\vec{J}((u\,z)_i).
\ee
This means that the $\U(N)$ transformations do not act on the space of (classical) polyhedra (with $N$ faces) but they truly act on the space of {\it framed  polyhedra} which have an extra $\U(1)$-phase attached to each face.

\subsection{A New Resolution of the Identity}
In this section we want to show that our $\U(N)$ coherent states form an over-complete basis. Since the space of label of our coherent state is the Grassmanian we need to describe first the integration measure $\Omega_{|z_{i}\ra}$ on $\mathrm{Gr}_{2,N}$.
Since $ \mathrm{Gr}_{2,N} = \C^{2N} /\mathrm{GL}(2,\C)$ we can relate this measure to the Lebesgue measure
on $\C^{2N}$ and $\mathrm{GL}(2,\C)$ as follows:
\be\label{Gint}
\int_{\C^{2N}} \left(\prod_i \rd^4|z_i \ra \right) e^{-\sum_{i}\la z_{i}|z_{i}\ra} F(|z_{i}\ra) =
\int_{\mathrm{GL}(2,\C)} \rd \lambda  e^{-\tr(\lambda \lambda^{\dag})}|\det(\lambda)|^{2(N-2)} \int_{\mathrm{Gr}_{2,N}} \Omega_{|z_{i}\ra} F(|z_{i}\ra)
\ee
for  $F$ a function invariant under $\mathrm{GL}(2,\C)$ transformations $|z_{i}\ra \to \lambda |z_{i}\ra$ and where
$\rd^{4}|z_{i}\ra\equiv\rd^{2} z_{i}^{0}\rd^{2}z_{i}^{1}$.
Explicitly this means that we can represent this measure as
\be
\Omega_{|z_{i}\ra} = \prod_{i}\rd^4|z_i \ra \, \delta^{(4)}\left(\sum_{i}|z_{i}\ra\la z_{i}|\right)
\ee
Taking $F=1$ in \Ref{Gint} we get that the volume of $ \mathrm{Gr}_{2,N}$ is given by $ V( \mathrm{Gr}_{2,N}) = \pi^{2(N-2)}/(N-1)!(N-2)!$
and we denote $\hat{\Omega}_{|z_{i}\ra}$ the normalised measure.
We now have the following proposition
\begin{prop}
These $\U(N)$ coherent states provide us with a new decomposition of the identity on the Hilbert space of intertwiners $\cH_N^{(J)}$:
\be\label{decu}
\id_N^{(J)} = D_{N,J} \int_{\mathrm{Gr}_{2,N}}
\frac{|J, z_{i})(J,z_{i}|}{(z_{i}|z_{i})^{J}}\,\,  \hat{\Omega}_{z_{i}}
\ee
where $\hat{\Omega}_{z_{i}}$ is the normalized measure on $\mathrm{Gr}_{2,N}$ and
$$
D_{N,J}=\dim \cH_N^{(J)}= \frac{(N+J-1)!(N+J-2)!}{J!(J+1)! (N-1)!(N-2)!}
$$
is the dimension of the $\U(N)$-representation  of weight $[J,J,0,\cdots,0]$.
\end{prop}

\begin{proof}
First,  we point out that the quotient $(z_{i}|z_{i})^{-J}\,|J, z_{i})(J,z_{i}|$ in the integral is well-defined on $\mathrm{GL}(2,\C)$ since it is invariant under $\mathrm{GL}(2,\C)$:
$$
|z_i\ra \arr\lambda\,|z_i\ra
\quad\Rightarrow\quad
(\lambda z_{i}|\lambda z_{i})\,=\,|\det(\lambda)|^2 (z_{i}|z_{i}),\quad
|J,\lambda\,z_i)( J,\lambda\,z_{i}|\,=\,|\det(\lambda)|^{2J} |J,z_i)(J,z_{i}|\,.
$$

Then, Let us call $\hat{\cal O}$ the operator on the right hand side of \Ref{decu}.
We know that  the coherent states behave consistently under $\U(N)$ transformations,
$u\,|J,z_i\ra\,=\,|J,(u\,z)_i\ra$,
and that the measure factors $(z_{i}|z_{i})$ and  $\Omega_{|z_{i}\ra}$ are invariant under the action of $\U(N)$.
This imply that $\hat{\cal O}$ commutes with all unitary transformations in $\U(N)$.
Since $\cH_N^{(J)}$ is an irreducible representation of $\U(N)$, we know that $\hat{\cal O}$ is proportional to the identity.

Then we only have to compute the trace of the operator to find out the proportionality coefficient and conclude the proof. This is obvious since  $(J,z_{i}|J, z_{i})=\,(z_{i}|z_{i})^{J}$.

\end{proof}

\subsection{Back to Coherent Intertwiners}

We would like to relate these $\U(N)$ coherent states to the usual  coherent intertwiner states previously introduced in \cite{coh1} and further developed in \cite{coh2,holomorph}. On the one hand, it would help clarify the semi-classical interpretation of these new $\U(N)$ coherent states. On the other hand, it could prove useful for future applications of the $\U(N)$ framework to spinfoam models since the standard coherent intertwiners are a key ingredient of the most recent EPRL-FK spinfoam models for 4d quantum gravity and their semi-classical analysis \cite{coh1, EPRL, FK, conrady,simone}.

Thus we introduce the coherent intertwiner states:
\be
||\vec{\jmath},z_{i}\ra \equiv \int_{\SU(2)} \rd g\,  g\triangleright \otimes_{i=1}^{N} |j_{i},z_{i}\ra.
\ee
The states $|j,z\ra \equiv  |z\ra^{\otimes 2j}$ are the $\SU(2)$ coherent states, which can be obtained by the action of creation operators on the vacuum state:
\be
|j,z\ra =\frac{\left(\zo a^{\dagger} + \zi b^{\dagger}\right)^{2j}}{\sqrt{(2j)!}} |0) .
\ee
Let us nevertheless point out the slight difference with the standard construction of $\SU(2)$ coherent states which are obtained by factorizing out $(\zo)^{2j}$ from the previous expression and writing the state solely in term of the ratio $\zi/\zo$ (see the appendix of \cite{UN} for more details).
The main purpose of this section is to shown that the coherent intertwinner are intimately related to the $U(N)$ invariant coherent states:
\begin{prop}\label{covint}
We have the following relation between $\U(N)$ coherent states and coherent intertwiners:
\be
\frac1{ J!\sqrt{({J+1})} } |J,z_{i})
\,=\,
|J,z_{i})_c \equiv \sum_{\sum{j_{i}}=J}  \frac{(-1)^J}{\sqrt{(2j_{i})!\cdots (2j_{N})!}} \,||\vec{\jmath},z_{i}\ra.
\ee
\end{prop}

From this result we can extract the scalar product between coherent intertwiners:
\begin{prop}
The scalar product between coherent intertwiner is given by
\be
\la \vec{\jmath},z_{i}||\vec{\jmath},w_{i}\ra = \frac{\prod_{i}(2j_{i})!}{(\sum_{i}j_{i} +1)!}\sum_{a_{ij}\in I_{\vec{\jmath}}}\,\, \prod_{i<j} \frac{(\bar{\Bz}_{ij}\Bw_{ij})^{a_{ij}}
}{a_{ij}!}
\ee
where
$I_{\vec{\jmath}}=\{a_{ij}\in \N| a_{ij}=a_{ji} \,\, \mathrm{and}\,\, \sum_{j:j\neq i}a_{ij} =2j_{i}\}$.
\end{prop}
\begin{proof}
The proof of the last proposition is straightforward. First one notice that the RHS is  homogeneous of degree $2j_{i}$ in $\bar{z_{i}}$ and $w_{i}$.
Then, if one sum the RHS divided by $\prod_{i}(2j_{i})!$ over all spins such that $\sum{j_{i}}=J$
one recover $(\sum_{i<j} \bar{\Bz}_{ij}\Bw_{ij})^{J}/J!(J+1)!$ which is  equal to ${}_{c}(J,z_{i}|J,w_{i})_{c}.$  according to proposition \ref{covint}.
\medskip

 We now go back to the proposition \ref{covint}.
To prove this relation, we shall compute the scalar product between the $\U(N)$ coherent states $|J,w_{i}\ra$ and the superposition of coherent intertwiners $|J,z_{i})_c$, for arbitrary spinors $w_i,z_i$.
First, using the fact that the operators $F_\Bw$ is invariant under $\SU(2)$, we can get rid of the gauge-averaging over $\SU(2)$ entering the definition of the coherent intertwiners $||\vec{\jmath},z_{i}\ra$. Then explicitly summing over the $j_i$ labels, we obtain the compact expression:
\be
\la J,w_{i}|J,z_{i})_c
\,=\,
\f1{J!(2J)!\sqrt{J+1}}\,( 0 | F_\Bw^J(K_z^\dag)^{2J}| 0),
\qquad\textrm{with}\quad
K_z^\dag=\sum_k^N \zo_k a_k^\dag+\zi_k b_k^\dag.
\ee
We compute the commutators:
\bea
[F_\Bw,K_z^\dag]&=&
\sum_{i,j} \overline{\Bw}_{ji}\,(\zi_j a_i-\zo_j b_i)
\,\equiv\, L \nn\\
{[}F_\Bw,L{]}&=& 0,\qquad
[L,K_z^\dag]=-\tr\, \Bw^\dag\Bz.
\eea
This allows us to calculate the commutator:
\bea
F_\Bw (K_z^\dag)^{2J}
&=&
(K_z^\dag)^{2J}F_\Bw +\sum_{n=0}^{2J-1} (K_z^\dag)^{2J-1-n}L(K_z^\dag)^n\nn\\
&=&
(K_z^\dag)^{2J}F_\Bw -(K_z^\dag)^{2J-2}\tr\,\Bw^\dag\Bz\,\sum_n^{2J-1} n
=
(K_z^\dag)^{2J}F_\Bw -J(2J-1)(K_z^\dag)^{2J-2}\tr\,\Bw^\dag\Bz\,.\nn
\eea
Iterating this relation, we finally compute:
\be
F_\Bw^J(K_z^\dag)^{2J}| 0\ra=\,
(2J)!\,\left(-\f12\tr\,\Bw^\dag\Bz\right)^J\,|0).
\ee
This shows that the equality between scalar products:
\be
( J,w_{i}|J,z_{i})_c
\,=\,
\f1{J!\sqrt{J+1}}\,\left(-\f12\tr\,\Bw\Bz^\dag\right)^J
\,=\,
\f{(-1)^J}{J!\sqrt{J+1}}\,( J,w_{i}|J,z_{i}).
\ee
Since the coherent states $|J,w_{i}\ra$ form a (over-)complete basis of states on the Hilbert space of intertwiners~\footnotemark, we can conclude to the equality of the states $|J,z_{i})_c$ and $|J,z_{i})$ (up to the explicit numerical factor appearing in the previous equation).

\footnotetext{
Explicitly, we can use the formula for the resolution of the identity in term of the $\U(N)$ coherent states:
$$
|J,z_{i})_c
\,=\,
D_{N,J} \int
\frac{|J, w_{i})( J,w_{i}|J,z_{i})_c}{ (w_{i}|w_{i})^{J}}\,\,  \Omega_{w_{i}}
\,=\,
\f{D_{N,J}}{J!\sqrt{J+1}} \int
\frac{|J, w_{i})( J,w_{i}|J,z_{i})}{(w_{i}|w_{i})^{J}}\,\,  \Omega_{w_{i}}
\,=\,
\f1{J!\sqrt{J+1}}\,|J,z_{i}).
$$
}

\end{proof}

It is interesting to give an alternative derivation of the scalar product between two $\U(N)$ coherent states by computing directly the scalar product between coherent intertwiners.
This derivation can be done in two steps: first one notice that
\bea
 {}_c\la J, z_{i}|J,w_i \ra_c &=& \int_{\SU(2)} \rd g
\sum_{\sum{j_{i}}=J}  \frac{\prod_{i}\la z_{i}|g| w_{i}\ra ^{2j_{i}}    }{{(2j_{i})!\cdots (2j_{N})!}}
 = \frac{1}{(2J)!}   \int_{\SU(2)}\rd g\,   \left(\tr  (X g)  \right)^{2J},
 \eea
where $X= \sum_{i}| w_{i}\ra \la z_{i}|$ is the same $2\times 2$ matrix introduced earlier.
Then one shows that
\be
 \frac{1}{(2J)!} \int_{\SU(2)}   \left(\tr  (X g)  \right)^{2J}  \rd g
 = \frac{\left(\mathrm{det}(X) \right)^{J}}{J! (J+1)!}.
 \ee
The explicit computation is done in appendix.
This confirms the equality between scalar products:
 \be
 {}_c\la J, z_{i}|J,w_{i} \ra_c =\frac{(J, z_{i}|J,w_{i} )}{ J! (J+1)!}.
 \ee

Let us add that the integrals over $\SU(2)$ computed above can be summarized in a single integral:
\bea
\int dg\, e^{i\f\lambda 2 \tr\,Xg}
&=&
\sum_n\f{(i\lambda)^n}{2^n n!}\int dg\,(\tr\,Xg)^n
\,=\,
\sum_J \f{(-1)^J}{J!\,(J+1)!} \left(\frac{\lambda^2\det X}4\right)^{J} \nn\\
&=&
\f{2J_1(\lambda\sqrt{\det X})}{\lambda\sqrt{\det X}}\,,
\eea
where $J_1$ is the Bessel function (of the first kind). We have used the obvious fact that $\int dg\,(\tr\,Xg)^n=0$ when $n$ is odd. When $X=i\vec{x}\cdot\vec{\sigma}$ is an element of the $\su(2)$-algebra depending on the vector $\vec{x}\in\R^3$, we recover the standard formula for the fuzzy $\delta_\star$-function for the $\star$-product dual to the convolution product over $\SU(2)$ \cite{star}~:
\be
\int dg\, e^{i\f\lambda 2 \tr\,Xg}
\,=\,
\int dg\, e^{-\f\lambda 2 \tr\,(\vec{x}\cdot\vec{\sigma})g}
\,=\,
\f{2\,J_1(\lambda|x|)}{\lambda |x|}
\,\equiv\,\delta_\star(\vec{x}).
\ee

%

%
%
%

\section{Conclusion}

We have pushed further the investigation of the $\U(N)$ structure of the intertwiner space initiated in \cite{OH,UN}.
Considering the space of all $N$-valent intertwiners, we decompose it separating the intertwiners with different total area $J\equiv \sum_i j_i$~:
$$
\cH_N\,\equiv\,\bigoplus_{\{j_i\}}\cH_{j_1,..,j_N}\,=\, \bigoplus_{J\in\N}\cH_N^{(J)}.
$$
In the previous work \cite{UN}, we have shown that each (sub)space $\cH_N^{(J)}$ at fixed total area carries an irreducible representation of $\U(N)$. The $\u(N)$ generators $E_{ij}$ are quadratic operators in the harmonic oscillators of the Schwinger representation of the $\su(2)$-algebra. In the present work, we established the full space $\cH_N$ as a Fock space by introducing annihilation and creation operators $F_{ij},F^\dag_{ij}$, which allow transitions between intertwiners with different total areas.

We further used these creation operators to define $\U(N)$ coherent states $|J,z_i\ra\propto \,F^\dag_\Bz\,|0)$ labeled by the total area $J$ and a set of $N$ spinors $z_i$ (up to multiplication by a global arbitrary complex number). These states turn out to have very interesting properties:
\begin{itemize}
\item They transform simply under $\U(N)$-transformations:
$u\,|J,z_i)\,=\,|J,(u\,z)_i)$.

\item They get simply rescaled under global $\GL(2,\C)$ transformation acting on all the spinors:
$$
|J,\lambda\,z_i) \,=\, (\det \lambda)^J\,|J,z_i).
$$
In particular, they are invariant under global $\SL(2,\C)$ transformations (among which $\SU(2)$ rotations) since they leave the matrix $\Bz$ invariant. Moreover, one can always map an arbitrary set of spinors $z_i$ onto a new set of spinors satisfying the closure constraints and with the same Pl\"u cker coordinates $\Bz_{ij}$.

\item They are coherent states {\it \`a la} Perelomov  and are obtained by the action of $\U(N)$ on highest weight states. These highest weight vectors correspond to bivalent intertwiners such as the state defined by $z_1=(z^0,z^1)$, $z_2=\varsigma z_1=(-\bar{z}^1,\bar{z}^0)$ and $z_k=0$ for ${k\ge 3}$.

\item For large areas $J$, they are semi-classical states peaked around the expectation values for the $\u(N)$ generators: $\la E_{ij}\ra= 2J\,\la z_i|z_j\ra/\sum_k \la z_k|z_k\ra$.

\item The scalar product between two coherent states is easily computed:
$$
(J,z_{i}|J,w_{i}) = \mathrm{det}\left(\sum_{i}|z_{i}\ra\la w_{i}|\right)^{J}
\,=\,
\left(\f12\tr \,\Bz\Bw^{\dagger}\right)^J.
$$

\item They are related to the coherent (and holomorphic) intertwiners already studied in \cite{coh1,coh2,holomorph}. Writing $||\vec{j},z_i\ra$ for the usual group-averaged tensor product of $\SU(2)$ coherent states defining the coherent  intertwiners, we have:
    $$
    |J,z_i)\propto \sum_{\sum j_i=J}
    \frac{(-1)^J}{\sqrt{(2j_{i})!\cdots (2j_{N})!}} \,||\vec{\jmath},z_{i}\ra.
    $$

\item The coherent states $|J,z_i\ra$ at fixed $J$ provide an over-complete basis on the space $\cH_N^{(J)}$. This gives a new decomposition of the identity on that space $\id_N^{(J)} = \int [d\mu(z_i)]\,|J, z_{i}) (J,z_{i}|$ where $[d\mu(z_i)]$ is a $\U(N)$-invariant measure (on $\C\mathbb{P}_{2N-1}$).

\end{itemize}

We believe that this $\U(N)$ framework opens the door to many applications in loop quantum gravity and spinfoam models. First, having semi-classical states coherent under the action $\U(N)$ should allow to understand better the geometric interpretation of these $\U(N)$ transformations in the large $N$-regime. A path would be how a classical 2-metric on the boundary dual to the intertwiner can be written in term of spinors $z$ on each point of that surface and study how the $\U(N=\infty)$ would act on such classical metric. This should allow to prove (or disprove) the relation between this $\U(N)$ action and area-preserving diffeomorphisms.

Second, the gauge-invariant operators $E$ and $F$ allow to separate the total area variable from deformations of the intertwiner at fixed area. This should be relevant to further understand the deformations (and thus ``diffeomorphisms") of spin network states in loop quantum gravity. To this purpose, we can use the operators $E$ and $F$ to write an algebra of gauge-invariant operators acting on spin network states (and not only on single intertwiners) \cite{soon}. Then we hope that this will allow to write spin network coherent states, most likely using the twisted geometry framework introduced by one of the authors \cite{twisted}. One issue to solve will be how to glue consistently the newly defined $\U(N)$ coherent states into spin network states.

Third, we can apply our $\U(N)$ coherent states to the interpretation and construction of spinfoam models. Indeed, the coherent intertwiners are a crucial ingredient of the standard spinfoam models \cite{EPRL,FK,simone}. Due to the explicit link between these coherent intertwiners and our new $\U(N)$ coherent states, we think that the $\U(N)$-framework could have a two-fold use. On the one hand, we can insert our new decomposition of the identity between quantized 4-simplices, understand how the (simplicty) constraints are expressed in term of $\u(N)$ generators and compute the spinfoam transition amplitudes in term of these new semi-classical states. On the other hand, we hope that using explicitly $\U(N)$ coherent states in the definition of spinfoam models will help to understand the symmetry of the spinfoam amplitudes and more particularly clarify the issue of discrete diffeomorphisms in the spinfoam framework.

\section{appendix}
Here we give the computation of \be
\frac{1}{(2J)!}   \int_{\SU(2)}\rd g\,   \left(\tr  (X g)  \right)^{2J}
\ee
We choose a specific parametrization of $\SU(2)$ group elements,
 \be
g(\alpha)=\left(\begin{array}{cc} \alpha_0 & -\bar{\alpha}_1 \\
                           \alpha_1 & \bar{\alpha}_0  \end{array}\right) \quad \mathrm{with} \quad \alpha_{0} = \sqrt{r} e^{i\phi}, \quad \alpha_{1}= \sqrt{(1-r)} e^{i\psi}.
\ee
The corresponding Haar measure is $ \rd g(\alpha) = \rd r \rd \phi \rd \psi /(2\pi)^{2}$ and
 the integral reads
 \bea\nonumber
& &  \,\,\,\frac{1}{(2J)!} \int_{\SU(2)}   \left(\tr  (X g)  \right)^{2J}  \rd g
=  \frac{1}{(2J)!} \int_{\SU(2)}   \left(\alpha_{0} X_{00} + \bar{\alpha}_{0} X_{11} + \alpha_{1} X_{01} - \bar{\alpha}_{1} X_{10}  \right)^{2J}  \rd g(\alpha) \\
&&= \sum_{a+b+c+d= 2J}  \frac{X_{00}^{a}}{a!} \frac{X_{11}^{b}}{b!} \frac{X_{01}^{c}}{c!}\frac{(-X_{10})^{d}}{d!}
 \int_{\SU(2)}   \alpha_{0}^{a} \bar{\alpha}_{0}^{b} \alpha_{1}^{c} \bar{\alpha}_{1}^{d}  \,\, \rd g(\alpha) \nonumber \\
&&= \sum_{a+b= J}  \frac{(X_{00} X_{11})^{a}}{(a!)^{2}}  \frac{(-X_{01} X_{10})^{b}}{(b!)^{2}}
\int_{0}^{1} r^{a} (1-r)^{b}
=  \sum_{a+b= J}  \frac{(X_{00} X_{11})^{a}}{(a!)^{2}}  \frac{(-X_{01} X_{10})^{b}}{(b!)^{2}} \frac{a! b!} { (a+b+1)!}
\nonumber \\
&&=\frac{\left( X_{00} X_{11}- X_{01} X_{10}\right)^{J}}{J!(J+1)!} \nonumber
= \frac{\left(\mathrm{det}(X) \right)^{J}}{J! (J+1)!}.
 \eea



\end{document}